\documentclass[usletter,12pt]{article}
\usepackage{basic_template}

\def\grad{\nabla}

\def\bb{\mathbf{b}}

\def\bx{\mathbf{x}}  
\def\by{\mathbf{y}}

\def\bo{\mathbf{0}}

\def\e{{\boldsymbol{\epsilon}}}

\def\n{{\boldsymbol{\nu}}}

\def\th{{\boldsymbol{\theta}}}

\def\b{{\boldsymbol{\beta}}}

\def\ph{{\boldsymbol{\phi}}}

\def\cB{\mathcal{B}}

\def\cE{\mathcal{E}}

\def\cG{\mathcal{G}}

\def\cL{\mathcal{L}}

\def\cS{\mathcal{S}}

\def\cX{\mathcal{X}}

\def\mC{\mathbb{C}}

\def\mR{\mathbb{R}}

\def\smskip{\smallskip}

\def\texitem#1{\par\smskip\noindent\hangindent 25pt
               \hbox to 25pt {\hss #1 ~}\ignorespaces}

\def\norm#1{\left\|#1\right\|}

\newcommand{\BEAS}{\begin{eqnarray*}}
\newcommand{\EEAS}{\end{eqnarray*}}
\newcommand{\BEA}{\begin{eqnarray}}
\newcommand{\EEA}{\end{eqnarray}}
\newcommand{\BEQ}{\begin{eqnarray}}
\newcommand{\EEQ}{\end{eqnarray}}
\newcommand{\BIT}{\begin{itemize}}
\newcommand{\EIT}{\end{itemize}}
\newcommand{\BNUM}{\begin{enumerate}}
\newcommand{\ENUM}{\end{enumerate}}

\newcommand{\BA}{\begin{array}}
\newcommand{\EA}{\end{array}}

\DeclareMathOperator*{\argmin}{argmin}

\def\prox#1{\mathbf{prox}_{#1}}

\newif\ifpagenumbering
\pagenumberingtrue

\pagenumberingfalse

\newsavebox{\theorembox}
\newsavebox{\lemmabox}
\newsavebox{\defnbox}
\newsavebox{\assbox}
\savebox{\theorembox}{\noindent\bf Theorem}
\savebox{\lemmabox}{\noindent\bf Lemma}
\savebox{\defnbox}{\noindent\bf Definition}
\usepackage[numbers]{natbib}
\usepackage{booktabs, tabularx,multicol, multirow}	
\usepackage{longtable}

\definechangesauthor[color=purple]{ID1}
\definechangesauthor[color=blue]{ID2}

\title{\textbf{Fitting ARMA Time Series Models without Identification: A Proximal Approach}}

\author[ ]{\textbf{Yin Liu}}
\author[ ]{\textbf{Sam Davanloo Tajbakhsh}}

\affil[ ]{Department of Integrated Systems Engineering 
\protect\\ The Ohio State University 
\protect\\
          \texttt{\{liu.6630, davanloo.1\}@osu.edu}}

\date{\today}

\begin{document}

\maketitle

\begin{abstract}
 Fitting autoregressive moving average (ARMA) time series models requires model identification before parameter estimation. Model identification involves determining the order of the autoregressive and moving average components which is generally performed by inspection of the autocorrelation and partial autocorrelation functions or other offline methods. In this work, we regularize the parameter estimation optimization problem with a non-smooth hierarchical sparsity-inducing penalty based on two path graphs that allow performing model identification and parameter estimation simultaneously. A proximal block coordinate descent algorithm is then proposed to solve the underlying optimization problem efficiently. The resulting model satisfies the required stationarity and invertibility conditions for ARMA models. Numerical results supporting the proposed method are also presented.
\end{abstract}

\section{INTRODUCTION}
\label{sec:intro}

ARIMA time series models have a multitude of applications, e.g., in epidemiological surveillance~\citep{zhang2014applications}, water resource management~\citep{wang2015improving}, transportation systems~\citep{billings2006application}, drought forecasting~\citep{han2010drought}, stock price forecasting~\citep{adebiyi2014comparison}, business planning~\citep{calheiros2014workload}, and power systems~\citep{chen2009arima}, to name a few. Even the emergence of deep neural networks and their customized architectures for time series modeling, e.g., Recurrent Neural Nets (RNN) and Long Short-Term Memory (LSTM) has not decreased the popularity of ARIMA models~\citep{makridakis2018statistical}.

Fitting ARMA($p,q$) time series models requires a two-step process: 1. Model identification, 2. Parameter estimation. The model identification step determines the order of the autoregressive (AR) component ($p$) and the moving average (MA) component ($q$). Next, given the underlying ARMA model, the parameters are estimated by solving an optimization problem for the maximum likelihood or least square estimates~\citep{box2015time,del2002statistical}. We should note that ARMA models are used to model stationary processes; however, there exists a more general class of ARIMA models for \emph{homogeneous nonstationary} processes (which are stationary in the mean). Such processes become stationary after $d$ times differencing; hence, the corresponding ARIMA($p,d,q$) model includes differencing of order $d$. The results of this paper are mainly for stationary processes with potential extension to the homogeneous nonstationary processes.  

Model identification is primarily based on visual inspection of the sample autocorrelation function (ACF) and partial autocorrelation (PACF) plots. For the AR($p$) process, the sample ACF follows an exponential decay, and the sample PACF cuts off after lag $p$, while for the MA($q$) process, the sample ACF cuts off after lag $q$ and the sample PACF decays exponentially~\citep{del2002statistical}. When the process involves both AR and MA components, it is more difficult to identify the correct orders. After model identification, parameters are estimated by minimizing a loss function (e.g., negative log-likelihood or least square). Some works, e.g., \citet{box2015time}, recommended an iterative approach between the model identification and parameter estimation which involves inspection of the residuals from the fitted model to make sure that they are indeed white noise. 

In many of today's applications, ARMA models should be fitted to many times series data $\{y_t^j\}_{j=1}^J$ with $J$ being very large, e.g., the demand data for more than thousands of products. If demand is uncorrelated across different products, fitting vector ARMA models is unnecessary, and separate models are preferable. In such scenarios, model identifications become a significant challenge in the modeling process. This work proposes a novel approach to fit ARMA models that allows automating the fitting procedure by merging the model identification step into the parameter estimation. Indeed, with the aid of a single tuning parameter, the proposed algorithm allows data to identify an appropriate model.

\subsection{CONTRIBUTIONS}
The contributions of this work are as follows:
\begin{itemize}
\item We develop a novel approach to fit ARMA time series models that identifies the model by tuning a single continuous parameter ($\lambda$). This approach merges model identification with parameter estimation by introducing a hierarchical sparsity-inducing penalty into the optimization problem. The sparsity-inducing penalty \emph{preserves the hierarchical model structure}, e.g., it does not allow the second AR parameter to be nonzero when the first AR parameter is zero. 
\item We propose an efficient proximal block coordinate descent (BCD) algorithm to solve the underlying nonsmooth and nonconvex optimization problem to a stationary point -- see Algorithm~\ref{alg:master_alg}. The proximal map of the nonsmooth hierarchical sparsity-inducing penalty is shown to be separable on the AR and MA components.
\item The proposed approach automates the ARMA time series modeling, does not require offline model identification and allows ARMA time series modeling for a large number of time series data. 
\end{itemize}

\subsection{Related Work}
Model identification to determine the order of the time series model through regularization with $\ell_1$-norm (also known as Lasso regularization) is performed for univariate AR models in \cite{wang2007regression, nardi2011autoregressive}. Extensions of such methods for vector AR (VAR) models are also considered in \cite{hsu2008subset}. By smart tuning of the regularization parameter, \citet{ren2010subset} proposed an adaptive Lasso regularizer for VAR models with provable asymptotic properties -- see also \cite{chan2011subset} for an adaptive ARMA model selection. This line of research utilizes $\ell_1$-penalty to induce sparsity in the parameters of the time series model to select a subset of the model parameters. However, naive usage of $\ell_1$-penalty results in models that lack the hierarchical structure. Hierarchically structured models are those in which higher-order parameters (in both the AR and MA components) are allowed to be nonzero when lower-order parameters are nonzero (as a necessary condition). This is similar to regression modeling where for better interpretability, one prefers to have higher-order interactions in the model only if the lower-level ones are included in the model.

To keep the benefits and simplicity of fitting ARMA models using Lasso-type penalties and also to enforce the desired hierarchical structure in the identified model, few works looked into hierarchical sparsity-inducing penalties for time series modeling. \citet{nicholson2014high} consider a hierarchical lag structure (HLag) for VAR models utilizing the group lasso with nested groups and use an iterative soft-thresholding algorithm to solve the underlying problem. Furthermore, \citet{wilms2017sparse} consider a vector ARMA model and propose to measure the complexity of the model based on a user-defined strongly convex function that can then be used as a regularizer for model identification. Their parameter estimation is a two-phase process: first, the unobservable errors are estimated by fitting a pure VAR($\infty$) model; next, the approximate lagged errors are used as the covariates for the MA component which results in a least-square problem regularized with $\ell_1$ or HLag penalty.

\subsection{Notations}
Lowercase boldface letters denote vectors, and uppercase Greek letters denote sets, except for $\cB$ which denotes the back-shift operator. The set of all real and complex numbers are denoted by $\mR$ and $\mC$, respectively. Given a set $g\subseteq\cG$, $|g|$ denotes its cardinality and $g^c$ denotes its complement. Given $\b\in\mR^d$ and $g\subseteq\{1,\cdots,d\}$, $\b_g\in\mR^{|g|}$ is a vector with its elements selected from $\b$ over the index set $g$.

\section{PROBLEM DEFINITION}\label{sec:problem_def}
We consider a stationary ARMA($p,q$) time series process with a zero mean as
\begin{equation}\label{eq:arma_1}
\begin{split}
y_t = \phi_1 y_{t-1}+\phi_2 y_{t-2}+\cdots+\phi_p y_{t-p}  - \theta_1\epsilon_{t-1}-\theta_2\epsilon_{t-2}-\cdots-\theta_q\epsilon_{t-q}+\epsilon_t,
\end{split}
\end{equation}
where $\phi_\ell$, with $\ell=1,\dots,p$ are the parameters of the AR component, and $\theta_\ell$, with $\ell=1,\dots,q$ are the parameters of the MA component, and $\epsilon_t$ is a white noise with zero mean and variance $\sigma^2$. The process \eqref{eq:arma_1} can also be written as
\begin{equation}\label{eq:arma_2}
P_\ph^p(\cB) y_t = P_\th^q(\cB)\epsilon_t,
\end{equation}
where $\cB$ is the \emph{back-shift operator}, i.e., $\cB y_t=y_{t-1}$, and 
\begin{equation}
P_{\boldsymbol{\alpha}}^d(z)\triangleq1-\alpha_1z-\alpha_2z^2-\cdots-\alpha_d z^d,
\end{equation}
is a polynomial of degree $d$ with the parameter $\boldsymbol{\alpha}$. The process \eqref{eq:arma_2} is \emph{stationary} if the AR component is stationary which is the case if all roots of the $P_\ph^p(z)$ polynomial are outside the unit circle; furthermore, the process is \emph{invertible} if the MA component is invertible which is the case if all roots of the $P_\th^d(z)$ polynomial are outside the unit circle~\citep{del2002statistical}. Requiring the two polynomials to have roots outside of the unit circle in the $\cB$ space translates to some constraints on $\ph=[\phi_1,\cdots,\phi_p]^\top$ and $\th=[\theta_1,\cdots,\theta_q]^\top$, i.e., $\ph\in\cX_{\ph}^p\subseteq\mR^p$ and $\th\in\cX_\th^q\subseteq\mR^q$, where $\cX_{\boldsymbol{\alpha}}^d$ is defined as
\begin{equation} \label{eq:opt}
\cX_{\boldsymbol{\alpha}}^d \triangleq \{\boldsymbol{\alpha}\in\mR^d:\ \forall z\in\mC, \ P_{\boldsymbol{\alpha}}^d(z)=0 \Rightarrow |z|>1 \}.
\end{equation}
We note that there is also another (maybe more common) representation for $\cX_{\boldsymbol{\alpha}}^d$ based on the monic polynomial
\begin{equation}
\bar{P}_{\boldsymbol{\alpha}}^d (z)\triangleq z^d+\alpha_1z^{d-1}+\cdots+\alpha_{d-1}z+\alpha_d,
\end{equation}
of degree $d$, where it can be shown that
\begin{equation} \label{eq:opt_2}
\cX_{\boldsymbol{\alpha}}^d = \{\boldsymbol{\alpha}\in\mR^d:\ \forall z\in\mC, \ \bar{P}_{\boldsymbol{\alpha}}^d(z)=0 \Rightarrow |z|<1 \}.
\end{equation}
Note that the new representation requires roots of the polynomial to be \emph{inside} the unit circle. For an arbitrary $d$, the geometrical complexity of $\cX_{\boldsymbol{\alpha}}^d$ makes projection onto this set very difficult \citep{combettes1992best}. Indeed, \citet{combettes1992best} discussed that $\cX_{\boldsymbol{\alpha}}^d$ is open, bounded, and \emph{not necessarily convex} --  see also \cite{moses1991determining, blondel2012explicit}. To deal with the openness of $\cX_{\boldsymbol{\alpha}}^d$, it is common to approximate it with a closed set from inside -- see \eqref{eq:apprx_cX}. However, projection onto this set or its approximation \emph{may not be unique} due to its potential nonconvexity. A method for projection onto the $\cX_{\boldsymbol{\alpha}}^d$ set was developed in \cite{moses1991determining}. While their scheme is easy to implement, the convergence of this iterative method is slower than the steepest descent method -- see also \cite{combettes1992best}. To conclude, imposing stationarity and invertibility of the model is performed by projecting $\ph$ and $\th$ onto (inner approximate of) $\cX_\ph^p$ and $\cX_\th^q$, respectively, which may not be unique.  

The above discussion is for an ARMA model that is already identified, i.e., $p$ and $q$ are known. For a model that is not identified, we also need
\begin{align}\label{eq:parent_child}
\begin{split}
 \bullet \quad & \text{if }   \phi_\ell=0 \text{ then } \phi_{\ell'}=0, \ \forall \ell<\ell',  \quad  \text{(or equivalently) }    \text{if } \phi_{\ell'}\neq0 \text{ then } \phi_{\ell}\neq0, \ \forall \ell<\ell', \\
\bullet \quad & \text{if }  \theta_\ell=0 \text{ then } \theta_{\ell'}=0, \ \forall \ell<\ell',   \quad \  \text{(or equivalently) }    \text{if } \theta_{\ell'}\neq0 \text{ then } \theta_{\ell}\neq0, \ \forall \ell<\ell', 
\end{split}
\end{align}
i.e., the sparsity of $\ph$ and $\th$ follow hierarchical structures.

Before discussing how these sparsity structures are enforced, we will briefly discuss the loss function for fitting ARMA models. Given an identified model, i.e., $p$ and $q$ are known, fitting ARMA models are generally performed by finding the conditional maximum likelihood or conditional least-square estimates, which are close to each other assuming that $\epsilon_t$ in \eqref{eq:arma_1} follows a Normal distribution and the data size $T$ is reasonably large. The conditional least-square estimate (for an identified model) requires solving
{\allowdisplaybreaks[0]
\begin{equation} \label{eq:cond_LS}
\begin{aligned}
\min_{\ph,\th}  \cL(\ph,\th) & = \frac{1}{2}\hspace{-0.4cm}\sum_{t=\max\{p,q\}}^T \hspace{-0.4cm}\hat{\epsilon}_t^2 = \frac{1}{2} \hspace{-0.4cm}\sum_{t=\max\{p,q\}}^T\hspace{-0.3cm}\left(y_t-\hat{y}_{t|t-1}(\ph,\th)\right)^2 \\                 \text{s.t.} & \quad \ph\in\cX_\ph^p, \quad \th\in\cX_\th^q, 
\end{aligned}
\end{equation}}
where $\hat{y}_{t|t-1}(\ph,\th)$ is the model prediction for $y_t$ using the data $\{y_t\}_{t=1}^{t-1}$, and is called \emph{conditional} since it depends on the $p$ initial values for $y_t$ and $q$ initial values for $\epsilon_t$. Note that in the absence of MA terms (i.e., $q=0$), the objective function of \eqref{eq:cond_LS} is convex in the parameters of the AR model $\ph$. However, if $q>0$ then the objective function of \eqref{eq:cond_LS} is also nonconvex, and optimization routines are not guaranteed to converge to the global optimum~\citep{hamilton1994time,box2015time,benidir1990nonconvexity,georgiou2008convex}. To sum up, in its most general case, problem \eqref{eq:cond_LS} involves nonconvex minimization over a nonconvex set and, hence, it is difficult to solve.

This paper intends to provide a solution that preserves the hierarchical sparsity structure and is \emph{not} concerned with the nonconvexities of the objective function and the feasible region. In the next section, we propose a method that allows learning $p$ and $q$ within the parameter estimation step.

\section{PROPOSED METHOD} \label{sec:prop_method}
Before discussing the proposed method, we should briefly discuss the notion of \emph{hierarchical sparsity}. Let $D=(\cS,\cE)$ be a Directed Acyclic Graph (DAG) where $\cS=\{s_1,\cdots,s_n\}$ is the set of graph nodes and $\cE$ be the set of ordered pair of nodes denoting edges where each pair denotes an edge from the node in the first element to the node in the second element. Each $s_i$ is an index set of the parameters of the model such that $s_i\cap s_j=\varnothing,\ \forall (i,j)$ and $\cup_{i=1}^n s_i=\{1,\cdots,d\}$ where $d$ is the number of parameters. DAG shows the sparsity structures of interest in the parent/child relationship. Assuming one variable per node, the variable in a child node can only be nonzero if the variable in the parent node is nonzero. For instance, given a parameter $\b\in\mR^3$, the top plot in Figure~\ref{fig:simple_path_graphs} requires $\beta_1\neq 0$ if $\beta_2\neq 0$ ($\beta_2=0$ if $\beta_1=0$); similarly, $\beta_2\neq 0$ if $\beta_3\neq 0$ ($\beta_3=0$ if $\beta_2=0$). 
For a DAG that contains more than one variable per node (e.g. the bottom plot in Figure~\ref{fig:simple_path_graphs}), two different hierarchies can be considered: 1. {\it Strong hierarchy}: the parameters in the child node can only be nonzero if \emph{all} of the parameters in its parent node(s) are nonzero. 2. {\it Weak hierarchy}: the parameters in the child node can be nonzero if \emph{at least one} of the parameters in its parent node(s) is nonzero~\citep{bien2013lasso}. For more information about hierarchical sparsity structures refer to \cite{zhao2009composite,jenatton2011proximal,jenatton2011structured,bach2012structured,yan2017hierarchical}.

\begin{figure}[hbt]
  \centering
  \includegraphics[scale=0.5]{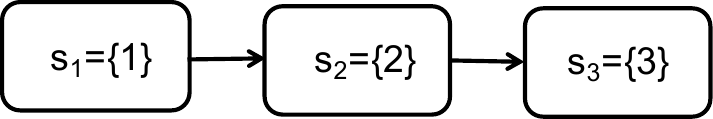}
  
  \hfill
  
   \includegraphics[scale=0.5]{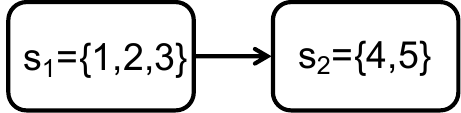}
  \caption{\footnotesize Path graphs showing hierarchical sparsities: \textbf{(Top)} A graph with a variable per node for $\b\in\mR^3$. \textbf{(Bottom)} A graph with multiple variables per node for $\b\in\mR^5$.}
    \label{fig:simple_path_graphs}
\end{figure}

\subsection{Hierarchical Sparsity for ARMA Models} \label{sec:hierarchical_arma}
In this work, we want to include the model identification of ARMA models in the parameter estimation step. We assume the knowledge about some upper bounds on the true $p^*$ and $q^*$, i.e., $\bar{p}\geq p^*$ and $\bar{q}\geq q^*$, respectively. Considering $\text{ARMA}(\bar{p},\bar{q})$, the estimated parameters should satisfy the condition \eqref{eq:parent_child}. To do so, we define two path graphs as shown in Figure~\ref{fig:arma_DAG}. Since this DAG consists of two path graphs and there is only one variable in each node, weak and strong hierarchies are equivalent. 
\begin{figure}[hbt]
  \centering
  \includegraphics[scale=0.6]{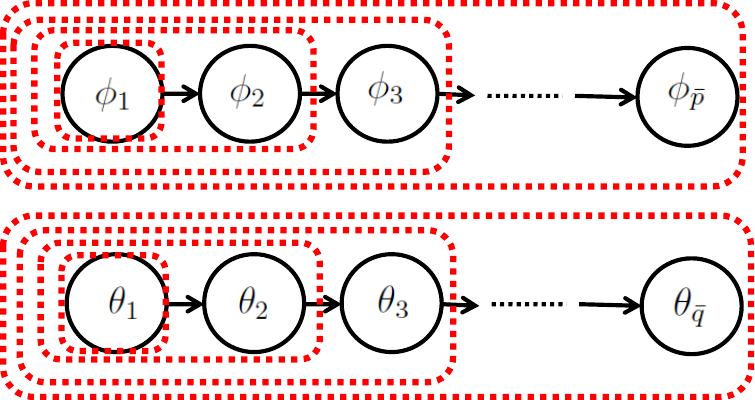}
  \caption{\footnotesize DAG for the $\text{ARMA}(\bar{p},\bar{q})$ process. The red dotted rectangles illustrate the ascending grouping scheme for the LOG penalty. }
    \label{fig:arma_DAG}
\end{figure}
Enforcing the sparsity structure shown in Figure~\ref{fig:arma_DAG} \emph{exactly} requires introducing binary variables into the optimization problem \eqref{eq:cond_LS} and solving a Mixed Integer Program (MIP). For instance, to model the parent/child hierarchy between $\phi_1$ and $\phi_2$, one needs to introduce a binary variable $z\in\{0,1\}$ and two constraints as $z\epsilon\leq|\phi_1|$ and $|\phi_2|\leq z\mu$ for some reasonably small and large parameters $\epsilon$ and $\mu$, respectively. Provided that the underlying optimization problem is already very difficult to solve, introducing $\bar{p}+\bar{q}-2$ binary variable makes the problem even more challenging. Hence, despite the significant recent advances in MIP algorithms (see e.g., \cite{manzour2019integer,bertsimas2016best,mazumder2017thediscrete,bertsimas2017sparse}), we use a convex nonsmooth regularizer that induces hierarchical sparsity structures of interest.

\subsection{Latent Overlapping Group (LOG) Lasso}
The hierarchical sparsity structure shown in Figure~\ref{fig:arma_DAG} is induced by regularizing the objective function in \eqref{eq:cond_LS} by the LOG penalty -- see \cite{jacob2009group}.  Let $\b\triangleq[\ph,\th]\in\mR^{(\bar{p}+\bar{q})}$ denote all of the ARMA parameters. The LOG penalty function is defined as
\begin{equation}\label{eq:LOG} 
\Omega_{\text{LOG}}(\b)= \inf_{\n^{(g)},\ g\in\cG} \left\{\sum_{g\in\cG}w_g\norm{\n^{(g)}}_2 \ \ \text{s.t.} \ \sum_{g\in\cG}\n^{(g)}=\b,\ \n^{(g)}_{g^c}=0 \right\},
\end{equation} 
where  \begin{align*}
\cG=\Big\{ & \{1\}, \{1,2\},\cdots,\{1,\cdots,\bar{p}\},\{\bar{p}+1\}, \\ 
 &\{\bar{p}+1,\bar{p}+2\}, \cdots,\{\bar{p}+1,\cdots,\bar{p}+\bar{q}\} \Big\},
\end{align*} $g\in\cG$ is itself a set, $\n^{(g)}\in\mR^{(\bar{p}+\bar{q})}$ is a latent vector indexed by $g$, and $w_g$ is the weight for the group $g$. $\n^{(g)}_{g^c}$ selects the elements of $\n^{(g)}$ based on the index $g^c$. The groups inside $\cG$ are shown with the red dotted rectangles in Figure~\ref{fig:arma_DAG}, i.e., for each node, there is a group containing this node and all of its ascendants.


It is known that $\ell_2$-norm induces block sparsity; hence, the LOG penalty tries to find block sparse combinations of the latent variables that sum up to $\b$~\citep{jacob2009group,yan2017hierarchical}. For instance, for an ARMA model with $\bar{p}=2$ and $\bar{q}=2$, $\cG=\{\{1\},\{1,2\},\{3\},\{3,4\}\}$, the objective of the infimum is $|\n^{\{1\}}_1|+\norm{[\n^{\{1,2\}}_1,\n^{\{1,2\}}_2]}+|\n^{\{3\}}_3|+\norm{[\n^{\{3,4\}}_3,\n^{\{3,4\}}_4]}$ (where for simplicity $w_g=1,\ \forall g\in\cG$) and the constraints are 
\begin{equation*}
\begin{bmatrix}
\n^{\{1\}}_1   \\
0                 \\
0 		   \\
0 		 
\end{bmatrix}
+
\begin{bmatrix}
\n^{\{1,2\}}_1  \\
\n^{\{1,2\}}_2  \\
0 		      \\
0 		      
\end{bmatrix}
+
\begin{bmatrix}
0 		  \\
0 		    \\
\n^{\{3\}}_3  \\
0		 
\end{bmatrix}
+
\begin{bmatrix}
0 			 \\
0 			  \\
 \n^{\{3,4\}}_3       \\
\n^{\{3,4\}}_4
\end{bmatrix}
=
\begin{bmatrix}
\phi_1	\\
\phi_2	 \\
\theta_1      \\
\theta_2
\end{bmatrix}.
\end{equation*}

\subsection{The Proposed Hierarchically Sparse Learning Problem}
The proposed Hierarchically Sparse (HS) learning problem is
\begin{equation} \tag{HS-ARMA}
	\label{eq:prop_opt}
\begin{split}
\min_{\ph,\th} \quad &\cL(\ph,\th) +\lambda\Omega_{\text{LOG}}(\ph,\th)	\\		
	\text{s.t.} & \quad \ph\in\cX_\ph^p, \quad \th\in\cX_\th^q,
\end{split}
\end{equation}
where $\lambda>0$ is a tuning parameter, $\cX_\ph^p$ and $\cX_\th^q$ are defined in \eqref{eq:opt}, and $\Omega_{\text{LOG}}(\cdot)$ is defined in \eqref{eq:LOG}. $\lambda$ controls the tradeoff between the loss and penalty functions and, hence, allows model identification and parameter estimation simultaneously. Increasing $\lambda$ results in sparser models where the resulted nested model satisfies the hierarchical sparsity structure shown in Figure~\ref{fig:arma_DAG}. As discussed in  Section~\ref{sec:hierarchical_arma}, $\bar{p}$ and $\bar{q}$ are some upper bounds on the true $p^*$ and $q^*$ and are known a priori. 

Given the convex nonsmooth function $\Omega_{\text{LOG}}(\cdot)$, we propose to solve \eqref{eq:prop_opt} using a proximal method~\citep{nesterov2013gradient,beck2009fast,parikh2014proximal}. Similar to gradient methods which require iterative evaluation of the gradient, proximal methods require iterative evaluation of the proximal operator. Proximal operator of the $\Omega_{\text{LOG}}(\bb)$ at $\bb\in\mR^{(\bar{p}+\bar{q})}$ is defined as
\begin{equation}\label{eq:prox_def}
\prox{\lambda\Omega_{\text{LOG}}}(\bb)\triangleq \argmin_{\b\in\mR^{(\bar{p}+\bar{q})}}\left\{\lambda\Omega_{\text{LOG}}(\b)+\frac{1}{2}\norm{\b-\bb}_2^2\right\}.
\end{equation}
\cite{zhang2020firstorder} developed a two-block alternating direction method of multiplier (ADMM) with a sharing scheme \citep{boyd2011distributed} to solve \eqref{eq:prox_def} -- see Algorithm~\ref{alg:sharing_Prox_LOG}. The proposed algorithm can be parallelized over all groups in $\cG$ in the update of the first block; furthermore, it converges \emph{linearly} -- see \cite{zhang2020firstorder} for more details.

\begin{algorithm}[htbp]\small
\caption{\small Evaluating $\prox{\lambda\Omega_{\text{LOG}}}(\bb)$}
\label{alg:sharing_Prox_LOG}
\begin{algorithmic}[1]
\REQUIRE $\bb, \lambda, \alpha, w_g \ \forall g\in\cG$
\STATE $k=0,\ U^0_{.g}=\bo,\ X^{2,0}_{.g}=\bo~\forall g\in\cG$
\WHILE{\text{stopping criterion not met}}
\STATE $k\gets k+1$ \\
\STATE $X_{gg}^{1,k+1} \gets \prox{\lambda w_g\norm{\cdot}_2}(X^{2,k}_{gg}-U^k_{gg}), \ \ \forall g\in\cG$  \\
\STATE $X_{g^cg}^{1,k+1} \gets \bo, \ \ \forall g\in\cG$
\STATE $\bar{\bx}^{2,k+1}\gets\frac{1}{|\cG|+\rho}\Big(\bb+\frac{\rho}{|\cG|}\sum_{g\in\cG}(X^{1,k+1}_{.g}+U^k_{.g})\Big)$ \\
\STATE  $X_{.g}^{2,k+1}\gets \bar{\bx}^{2,k+1}+X^{1,k+1}_{.g}+U^k_{.g}-(1/|\cG|)\sum_{g\in\cG}(X^{1,k+1}_{.g}+U^k_{.g}),\ \forall g\in\cG$ \\
\STATE $U_{.g}^{k+1}=U_{.g}^{k}+(\alpha/\rho)\big(\frac{1}{|\cG|}\sum_{g\in\cG}(X^{1,k+1}_{.g}+U^k_{.g})-\bar{\bx}^2\big), \ \forall g\in\cG.$ \\
\ENDWHILE\\
\STATE $\b=\sum_{g\in\cG} X^{1,k+1}_{.g}$ \\
\hspace{-0.7cm } \textbf{Output:} $\b$
\end{algorithmic}
\end{algorithm}

Let $\Omega_{\text{LOG}}^{\text{AR}}$  and $\Omega_{\text{LOG}}^{\text{MA}}$ be the LOG penalties for the pure AR, i.e, ARMA$(\bar{p},0)$, and pure MA, i.e, ARMA$(0,\bar{q})$, models, respectively. In Lemma~\ref{lem:separable_prox} below, we show that the proximal operator of $\Omega_{\text{LOG}}$ is separable over $\ph$ and $\th$.

\begin{lemma}\label{lem:separable_prox}
The proximal operator of the LOG penalty defined over the ARMA DAG is separable, i.e., $\prox{\Omega_{\text{LOG}}}(\bb_1,\bb_2)=( \prox{\Omega_{\text{LOG}}^{\text{AR}}}(\bb_1),\prox{\Omega_{\text{LOG}}^{\text{MA}}}(\bb_2) )$.
\end{lemma}

\begin{proof}
With a slight abuse of notation, let $\cG^{\text{AR}}$ be the set of groups for $\Omega_{\text{LOG}}^{\text{AR}}$ such that $\sum_{g\in\cG^{\text{AR}}}\n^{(g)} = \ph$ (the top path graph in Figure~\ref{fig:arma_DAG}). Similarly,  let $\cG^{\text{MA}}$ be the set of groups for $\Omega_{\text{LOG}}^{\text{MA}}$ such that $\sum_{g\in\cG^{\text{MA}}}\boldsymbol{\omega}^{(g)} = \th$. Given that the objective of the infimum in the definition of $\Omega_{\text{LOG}}$ for the ARMA DAG is separable in $\cG^{\text{AR}}$ and $\cG^{\text{MA}}$, we have $\Omega_{\text{LOG}}(\ph,\th)=\Omega_{\text{LOG}}^{\text{AR}}(\ph)+\Omega_{\text{LOG}}^{\text{MA}}(\th)$. Hence, the result follows from the separable sum property of the proximal operator.
\end{proof}

Indeed, in Algorithm~\ref{alg:master_alg}, the proximal operator of LOG is not evaluated in one step while the algorithm evaluates $\prox{\lambda\Omega_{\text{LOG}}^{\text{AR}}}$ and $\prox{\lambda\Omega_{\text{LOG}}^{\text{MA}}}$ sequentially in a Gauss-Seidel manner.

\begin{algorithm}[htbp]\small 
\caption{\small Proximal BCD to solve \eqref{eq:prop_opt}}
\label{alg:master_alg}
\begin{algorithmic}[1]
\REQUIRE $\lambda, \bar{p}, \bar{q}, \ph_0\in\cX_{\ph}^{\bar{p}}, \th_0\in\cX_{\th}^{\bar{q}}$
\STATE $k=1$  \\
\WHILE{\text{stopping criterion not met}}
\STATE $\ph^{k+1/2} \gets \prox{\lambda\Omega_{\text{LOG}}^{\text{AR}}}(\ph^k-\gamma_k\grad_{\ph}\cL(\ph^k,\th^k))$ \quad \text{(prox is calculated by Algorithm~\ref{alg:sharing_Prox_LOG})}\\
\STATE $p \gets \text{card}(\{i: \ph^{k+1/2}_i\neq 0\})$ \\
\STATE $\ph^{k+1} \gets \text{Proj}_{\tilde{\cX}_\ph^p}(\ph^{k+1/2})$ \\
\STATE $\th^{k+1/2} \gets \prox{\lambda\Omega_{\text{LOG}}^{\text{MA}}}(\th^k-\gamma_k\grad_{\th}\cL(\ph^{k+1},\th^k))$ \quad \text{(prox is calculated by Algorithm~\ref{alg:sharing_Prox_LOG})} \\
\STATE $q \gets \text{card}(\{i: \th^{k+1/2}_i\neq 0\})$\\
\STATE $\th^{k+1} \gets \text{Proj}_{\tilde{\cX}_\th^q}(\th^{k+1/2})$ \\
\STATE $k\gets k+1$ \\
\ENDWHILE\\
\hspace{-0.7cm } \textbf{Output:} $(\ph_k,\th_k)$
\end{algorithmic}
\end{algorithm}

The algorithm to solve problem \eqref{eq:prop_opt} is a two-block proximal block coordinate descent (BCD) with projection, shown in Algorithm~\ref{alg:master_alg}.
From \eqref{eq:arma_1}, since $\epsilon_t=y_t-\ph^\top\by_{t-p}^{t-1}-\th^\top\e_{t-q}^{t-1}$ where $\by_{t-p}^{t-1}=[y_{t-1},\cdots,y_{t-p}]$ and $\e_{t-q}^{t-1}=[\epsilon_{t-1},\cdots,\epsilon_{t-q}]$, we have
\begin{subequations}
\begin{align}
\grad_{\ph} \cL(\ph,\th) &= -\hspace{-0.4cm}\sum_{t=\max\{\bar{p},\bar{q}\}}^T \hspace{-0.4cm} (y_t-\ph^\top\by_{t-p}^{t-1}-\th^\top\e_{t-q}^{t-1})\by_{t-p}^{t-1}, \label{eq:grad_wrt_phi} \\
\grad_{\th} \cL(\ph,\th) &= -\hspace{-0.4cm}\sum_{t=\max\{\bar{p},\bar{q}\}}^T (y_t-\ph^\top\by_{t-p}^{t-1}-\th^\top\e_{t-q}^{t-1})\e_{t-q}^{t-1}. \label{eq:grad_wrt_theta}
\end{align}
\end{subequations}
The gradient updates are passed to proximal operators as arguments which are indeed proximal gradient steps~\citep{beck2009fast,parikh2014proximal}. Note that the solution of the proximal operators is sparse vectors that conform to the hierarchical sparsity of Figure~\ref{fig:arma_DAG}.  

The solutions of the proximal gradient steps for the AR and MA components, i.e., $\ph^{k+1/2}$ and $\th^{k+1/2}$ are not necessarily stationary or invertible, respectively. The stationarity and invertibility of AR and MA are regained by projection on $\cX_{\ph}^p$ and $\cX_{\th}^q$ where $p$ and $q$ are the order of AR and MA components from the proximal steps. For the projection, we use the second definition of $\cX_{\boldsymbol{\alpha}}^d$ in \eqref{eq:opt_2}. Since $\cX_{\boldsymbol{\alpha}}^d$ is an open set, following \cite{combettes1992best}, we find its approximation with a closed set from inside as
\begin{equation}\label{eq:apprx_cX}
\begin{split}
\tilde{\cX}_{\boldsymbol{\alpha}}^d(\delta) \triangleq & \{\boldsymbol{\alpha}\in\mR^d:\ \forall z\in\mC, \ \bar{P}_{\boldsymbol{\alpha}}^d(z)=0 \\
& \quad \Rightarrow -1+\delta\leq z \leq1-\delta \},
\end{split}
\end{equation}
where $\delta>0$ determines the approximation gap. Euclidean projection on $\tilde{\cX}_{\ph}^p(\delta)$ and $\tilde{\cX}_{\th}^q(\delta)$ sets guarantee stationarity and invertibility of $\ph^{t+1}$ and $\th^{t+1}$, respectively. Note that these projections do \emph{not} change the sparsity of the parameters.

Finally, note that $\epsilon_t$ in the objective of \eqref{eq:prop_opt} is calculated based on ARMA$(\bar{p},\bar{q})$. Hence, while the iterates $\ph^{t+1}$ and $\th^{t+1}$ are feasible with respect to $\cX_{\ph}^p$ and $\cX_{\th}^q$, respectively, we need to show $(\ph^{k+1},\th^{k+1})\in\cX_\ph^{\bar{p}}\times\cX_\th^{\bar{q}}$. This is established in Lemma~\ref{lem:nested} below.

\begin{lemma}\label{lem:nested}
For any $d\in\{1,2,...\}$, we have $\cX_{\boldsymbol{\alpha}}^d\subseteq\cX_{\boldsymbol{\alpha}}^{d+1}$. 
\end{lemma}
\vspace{-0.3cm}
\begin{proof}
Proof follows from the definition of $\cX_{\boldsymbol{\alpha}}^d$ in \eqref{eq:opt}, and that if $\boldsymbol{\alpha}\in\cX_{\boldsymbol{\alpha}}^d$ then $[\boldsymbol{\alpha},0]\in\cX_{\boldsymbol{\alpha}}^{d+1}$.
\end{proof}
Therefore, $\{\cX_{\boldsymbol{\alpha}}^d\}_{d=1}^{\bar{d}}$ is a sequence of nested sets as $\cX_{\boldsymbol{\alpha}}^1\subseteq\cdots\subseteq\cX_{\boldsymbol{\alpha}}^{\bar{d}}$. However, the reverse is not true, i.e., $\boldsymbol{\alpha}\in\cX^d_{\boldsymbol{\alpha}}$ is \emph{not sufficient} for $[\alpha_1,\cdots,\alpha_{d-1}]\in\cX^{d-1}_{\boldsymbol{\alpha}}$, which can be shown by counter examples.

\subsection{A Note on the Optimization Problem (HS-ARMA)}\label{sec:note}
Problem \eqref{eq:prop_opt} requires nonconvex and non-smooth optimization over a nonconvex set. To be specific, if $q=0$ the loss function is convex in $\ph$; otherwise, $\cL(\ph,\th)$ is nonconvex in \emph{both} $\ph$ and $\th$. Indeed, when $q>0$ the objective function is a \emph{polynomial} function of degree $T-\max\{p,q\}$. The $\Omega_{\text{LOG}}(\ph,\th)$ penalty is jointly convex but nonsmooth unction. Finally, $\cX_{\ph}^p$ and $\cX_{\th}^q$ are open nonconvex sets and their approximations $\tilde{\cX}_{\ph}^p$  and $\tilde{\cX}_{\th}^q$ (defined in \eqref{eq:apprx_cX}) are closed but still nonconvex sets. 

To deal with nonconvexities of $\tilde{\cX}_{\ph}^p$  and $\tilde{\cX}_{\th}^q$, one may try to approximate them with some inscribed convex sets which require generalizations of the \emph{potato peeling problem} \citep{goodman1981largest} and the algorithm in \cite{chang1986polynomial} to non-polygon geometries -- see also \cite{cabello2017peeling}. Note that optimization over the convex hulls of these sets may result in nonstationary or noninvertible solutions. 

Under some convex approximations of the sets $\tilde{\cX}_{\ph}^p$  and $\tilde{\cX}_{\th}^q$, the problem under investigation is a nonconvex nonsmooth optimization over a convex set. For such a setting, algorithms are settled with finding solutions that satisfy some necessary optimality conditions, e.g., stationary solutions which are those that lack a feasible descent direction. To the best of our knowledge, the only study that provides a method that converges to stationary points in this setting is \cite{razaviyayn2013unified}, which involves iterative minimization of a consistent majorizer of the objective function over the feasible set.

\section{NUMERICAL STUDIES}
\label{sec:numerical}
The corresponding code is provided in \url{https://github.com/Yin-LIU/ARMA_identify_proximal}.

\subsection{Synthetic Data Generation Process}\label{sec:data_gen}
To generate a stationary and an invertible ARMA($p^*,q^*$) model,  we first generate $p^*+q^*$ numbers uniformly at random on $[-1,-0.1]\cup [0.1,1]$ for all parameters. The samples are then rejected if the stationary and invertibility conditions, based on \eqref{eq:opt_2}, are not satisfied. Given an accepted sampled parameter $(\ph^{*,i},\th^{*,i})$, a realization of the time series with length $T=4000$ is simulated with a zero mean and variance equal to one.

\subsection{Model Identification and Parameter Estimation Accuracy}

\begin{table*}[t]
 \caption{ The mean (standard deviation) of HS-ARMA estimation errors. Boldface numbers are the minimum mean error for each model (row). Parameter estimates are obtained by the proximal BCD Algorithm~\ref{alg:master_alg}.}
  \centering
  \resizebox{0.8\textwidth}{!}{
    \begin{tabular}{ccccccc}
      \toprule
      \multicolumn{1}{c}{\multirow{2}[0]{*}{($p^*, q^*$)}} & \multicolumn{6}{c}{$\lambda_0$}                                                                                                                          \\
                                                       & \multicolumn{1}{c}{0.5}         & \multicolumn{1}{c}{1} & \multicolumn{1}{c}{2} & \multicolumn{1}{c}{3} & \multicolumn{1}{c}{5} & \multicolumn{1}{c}{10} \\
      \cmidrule(r){2-7}
      (3,2)                                            & 0.62 (0.350)                     & \textbf{0.38} (0.392)  & 0.54 (0.451)           & 0.59 (0.438)           & 0.63 (0.394)           & 0.77 (0.338)            \\
      (3,3)                                            & \textbf{0.64} (0.494)            & 0.74 (0.518)           & 0.85 (0.491)           & 0.92 (0.486)           & 1.05 (0.393)           & 1.05 (0.355)            \\
      (2,6)                                            & 0.79 (0.326)                     & 0.58 (0.324)           & \textbf{0.46} (0.339)  & 0.64 (0.368)           & 0.92 (0.390)           & 1.04(0.435)            \\
      (6,6)                                            & \textbf{0.69} (0.414)            & 0.79 (0.518)           & 1.10 (0.477)           & 1.25 (0.468)           & 1.32 (0.420)           & 1.29 (0.344)            \\
      (8,5)                                            & \textbf{0.87} (0.307)            & 0.99 (0.426)           & 1.22 (0.492)           & 1.41 (0.586)           & 1.57 (0.492)           & 1.48 (0.502)            \\
      \bottomrule
    \end{tabular}
  }
  \label{tab:average_error}
\end{table*}

To evaluate the estimation error of the proposed method, we simulate $n=20$ realizations of ARMA models with orders $(p^*,q^*)$ such that $p^*\leq \bar{p}=10$ and $q^*\leq\bar{q}=10$ following our discussion in Section~\ref{sec:data_gen}. The tuning parameter of the $\Omega_{\text{LOG}}$ penalty is set as $\lambda = \lambda_0 \sqrt{T}$ with $\lambda_0 \in \{0.5,1,2,3,5,10\}$ and $w_g$ in its definition is set to $|g|^{1/2}$. The estimation error is calculated as $ \epsilon_{\lambda_0} = \|(\hat{\ph}_{\lambda_0},\hat{\th}_{\lambda_0})-(\ph^{*},\th^{*}) \|_2$, where $(\ph^*,\th^*)$ are the true and $(\hat{\ph},\hat{\th})$ are the estimated parameters based on Algorithm~\ref{alg:master_alg}. Table \ref{tab:average_error} reports the mean and standard deviation of the estimation errors for different $\lambda_0$ values.

To provide a better understanding of the quality of parameter estimates and how they conform to the induced sparsity structure in Figure~\ref{fig:arma_DAG}, we conducted another study. First, we sampled one realization from 10 different ARMA(3,2) models. Then, with $\bar{p}=\bar{q}=5$ and $\lambda_0 \in \{0.5,1,2,3,5,10\}$, the HS-ARMA parameter estimates $(\hat{\ph}^i_{\lambda_0}, \hat{\th}^i_{\lambda_0})$ are calculated using Algorithm~\ref{alg:master_alg} and reported along with the true parameters $(\ph^{*,i},\th^{*,i})$ in Table \ref{tab:sample_detail} in Appendix A, where $i$ is the simulation index. Simple tuning of $\lambda_0$ allows the method to correctly identify the true orders $(p^*,q^*)$ and the estimated parameters conform to the underlying sparsity structure. Furthermore, the estimation errors are reasonably small. We also compared the estimation errors with pre-identified models where their parameters are estimated using a package -- see Figure \ref{fig:distance}. The mean of the HS-ARMA estimation error lies between those of the correctly and incorrectly identified (by one order in the AR component) models. For some samples with $\lambda_0$ around 2 or 3, the error of HS-ARMA is very close to the correctly identified ARMA model.

\begin{figure}[hbt]
  \centering
  \includegraphics[width=0.4\textwidth]{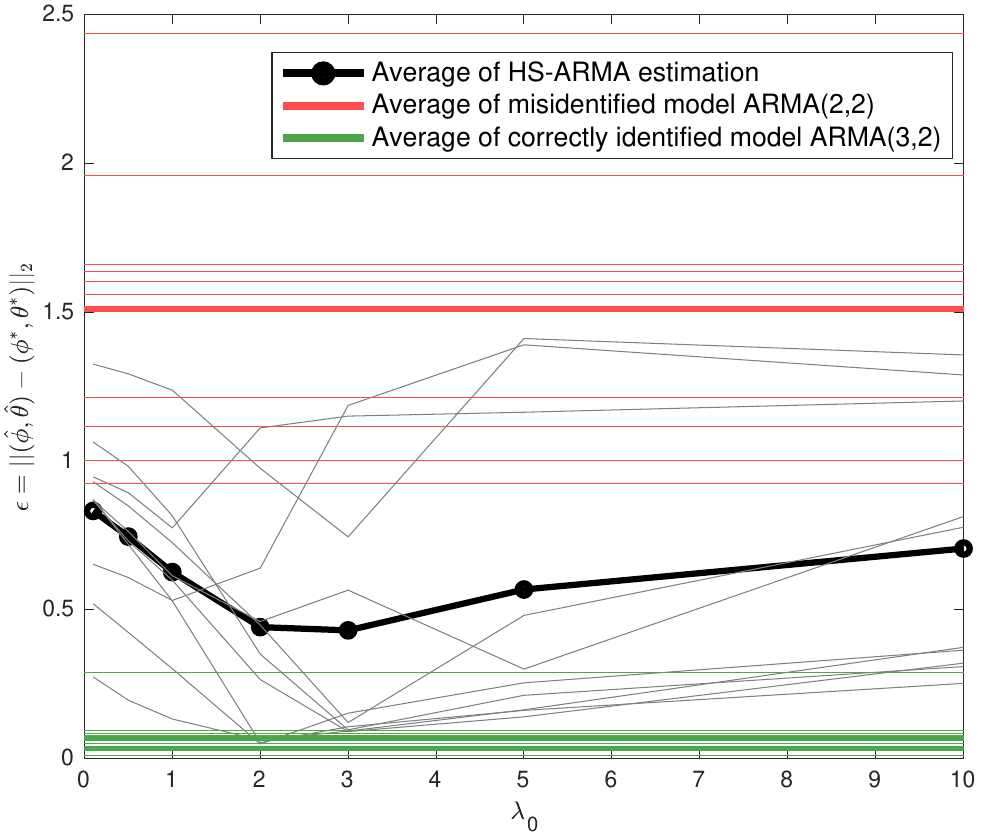}
  \caption{\footnotesize The estimation error of HS-ARMA and two pre-identified models. The three thicker lines are the mean estimation errors and the thinner lines represent estimation errors for each sample.}
  \label{fig:distance}
\end{figure}

\subsection{Prediction Performance}
We also compare the prediction performance of the HS-ARMA with those of correctly and incorrectly identified models using 10 realizations of one ARMA(3,2) model. For each realization, the estimated parameters with $\lambda_0\in \{0.5,1,2,3,5\}$ are used to forecast the process for the next 20 time points. Note that $\lambda_0=10$ is omitted because the fitted parameters were too sparse. Figure \ref{fig:prediction} illustrates the Root Mean Square Error (RMSE) for these methods.

For some $\lambda_0$, the RMSE of HS-ARMA is smaller than that of the correctly identified ARMA model. Furthermore, all HS-ARMA predictions for different $\lambda_0$ values have significantly lower RMSE compared to the incorrectly identified model.

\begin{figure}[hbt]
  \centering
  \includegraphics[width=0.4\textwidth]{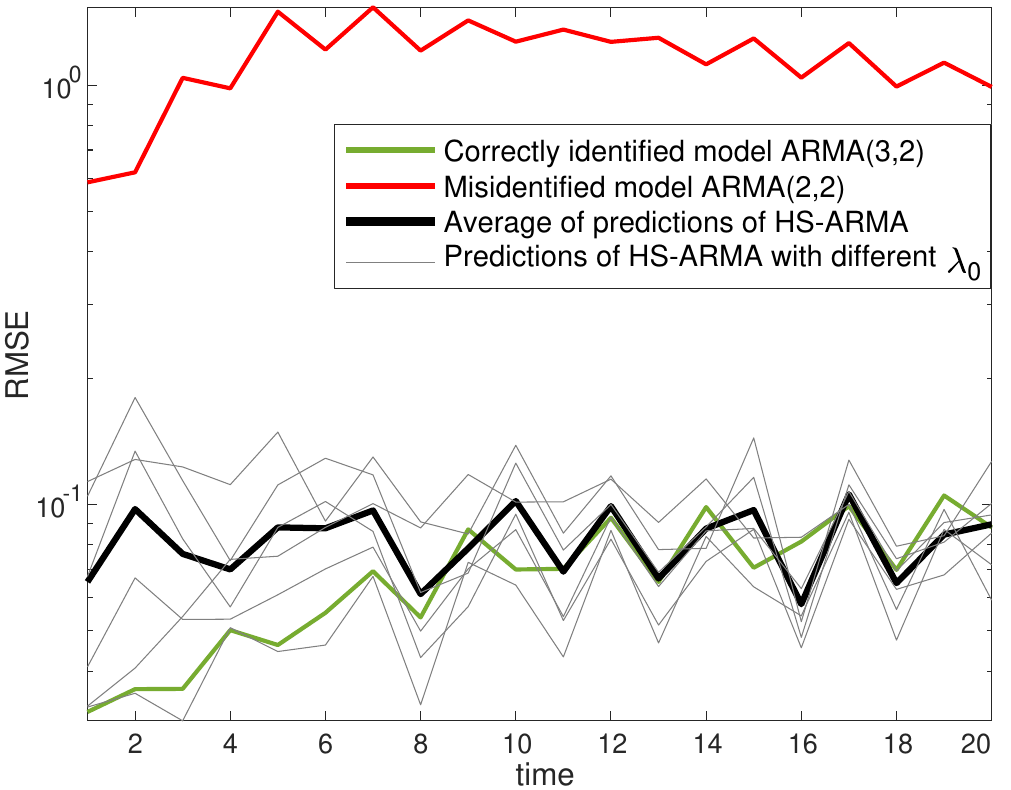}
  \caption{\footnotesize Prediction RMSEs for the HS-ARMA method vs. the correctly and incorrectly identified models. Each grey thin line is the RMSE of HS-ARMA with one $\lambda_0$ from $\{0.5,1,2,3,5\}$ from ten realizations and the black thick line is the average of the grey lines. The green and red lines are the RMSEs from the ten realizations for correctly and incorrectly identified models.}
  \label{fig:prediction}
\end{figure}

\subsection{Comparison with Other Methods}
We also compare our method with the one proposed by \cite{wilms2017sparse} (the ``bigtime'' R package) which also considers the hierarchical lag penalty, namely H-Lag penalty, and the lasso $\ell_1$ penalty. While the underlying optimization problem and the algorithm to solve it are fundamentally different than those proposed in this work, we believe the method in \cite{wilms2017sparse} specifically with the H-Lag penalty is the best benchmark for the proposed method. The parameter estimation method in \cite{wilms2017sparse} has two different phases. In the first phase, their method estimates a pure AR model, since every invertible ARMA process can be represented by an AR($\infty$) model. The estimated AR model is used to approximate the unobserved error terms which are used as the covariates of the MA component. In the second phase, a least-square objective regularized with a sparsity-inducing penalty is used to estimate the parameters of the final model.

We consider 9 different combinations of ARMA($p,q$) models with $p,q\in\{2,5,8\}$. For each combination, we construct four scenarios where the maximum absolute value of the roots of AR and MA components are chosen to be either 0.5 (invertible/stationary process) or 0.99 (close to none invertible/stationary process). In each scenario, we randomly generate an ARMA model and simulate 20 time series of length 200. Following the setting in \cite{wilms2017sparse}, the maximum potential lag $\bar{p}=\bar{q} =\lfloor0.75\sqrt{(T)}\rfloor = 10$, and AR and MA penalty coefficients $\lambda_p$ and $\lambda_q$ belong to 10 logarithmically spaced points between 1 and 100. The best combination of $\lambda_p$ and $\lambda_q$ parameters are determined by Bayesian Information Criterion (BIC). For each combination of $(p,q)$, the RMSEs are averaged over the final selected models. The results are presented in Figure~\ref{fig:compare}. In most cases, H-Lag penalty has lower RMSEs compared to the $\ell_1$ penalty; however, the proposed HS-ARMA method has the lowest RMSEs in all 9 cases.

\begin{figure}[hbt]
  \centering
  \includegraphics[width=0.4\textwidth]{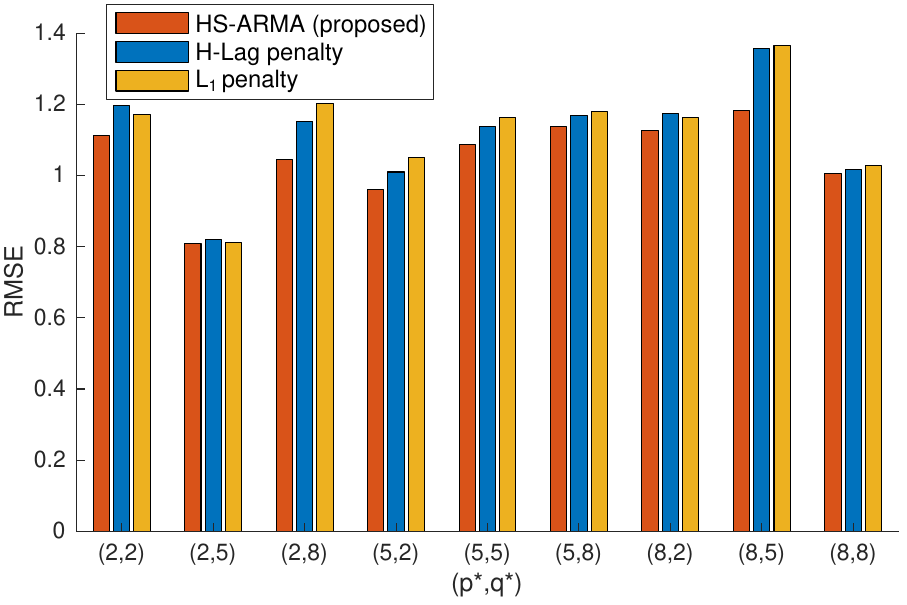}
  \caption{Comparison of the proposed HS-ARMA method with the hierarchical lag (H-Lag) and $\ell_1$ penalty methods from \cite{wilms2017sparse}.}
  \label{fig:compare}
\end{figure}

\subsection{Real Time Series Prediction}
We also implement the proposed HS-ARMA method on the real dataset, which is the Netflix stock prices from 4/15/2020 to 9/24/2021 with a total of 355 data points. To evaluate the influence of the penalty parameter of  HS-ARMA, the model is fitted with different combinations of $\lambda_{AR}$ and $\lambda_{MA}$. After the ARMA model is identified by HS-ARMA, we evaluate the performance of this model by different criteria, including AIC, AICc, and BIC.   The results are presented in Figure \ref{fig:netflix}. It is obvious that a larger penalty parameter will enforce the lower order of the ARMA model and the best model occurs when the value of the parameter is set properly. The advantage of HS-ARMA is that the search in the continuous $\lambda$ parameter space can be performed more efficiently compared to a brute-force grid search. For instance, the proposed algorithm can be easily incorporated into a hyperparameter optimization method (e.g., \cite{franceschi2018}) with a polynomial time solution to find the optimal $\lambda$ values.

\begin{figure*}
\vspace{-0.4cm}
    \centering
    \begin{subfigure}[t]{0.33\textwidth}
        \centering
        \includegraphics[width=\textwidth]{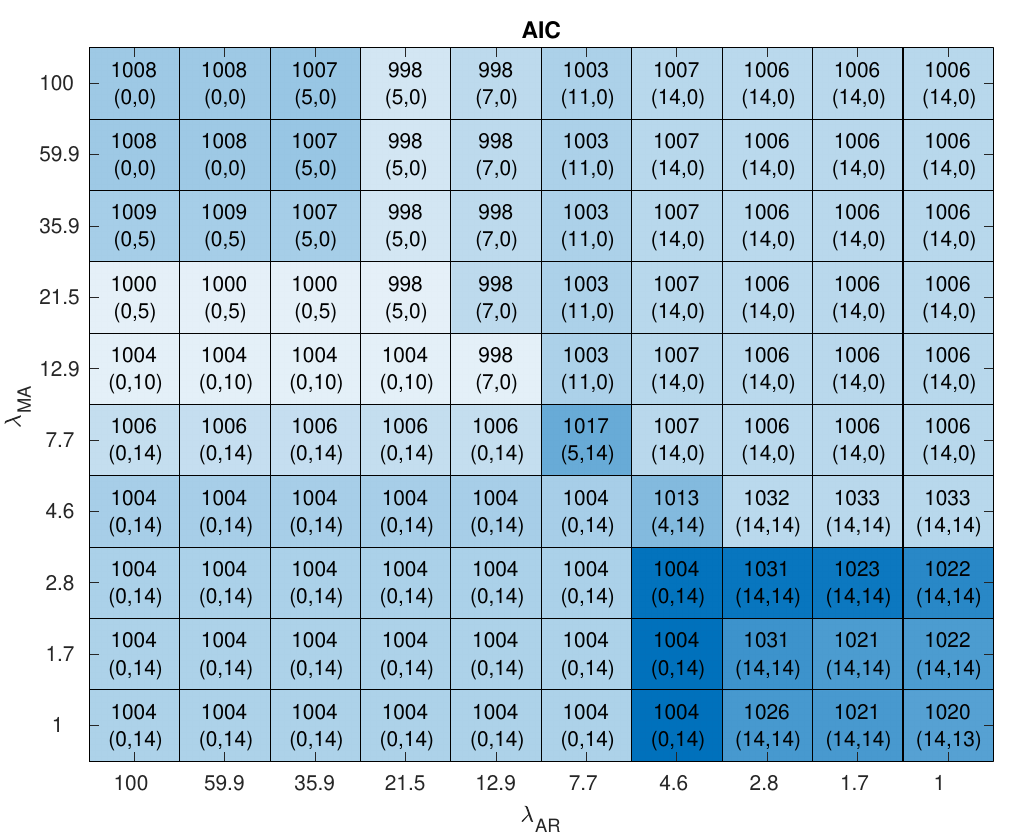}
    \end{subfigure}%
    \hfill
    \begin{subfigure}[t]{0.33\textwidth}
        \centering
        \includegraphics[width=\textwidth]{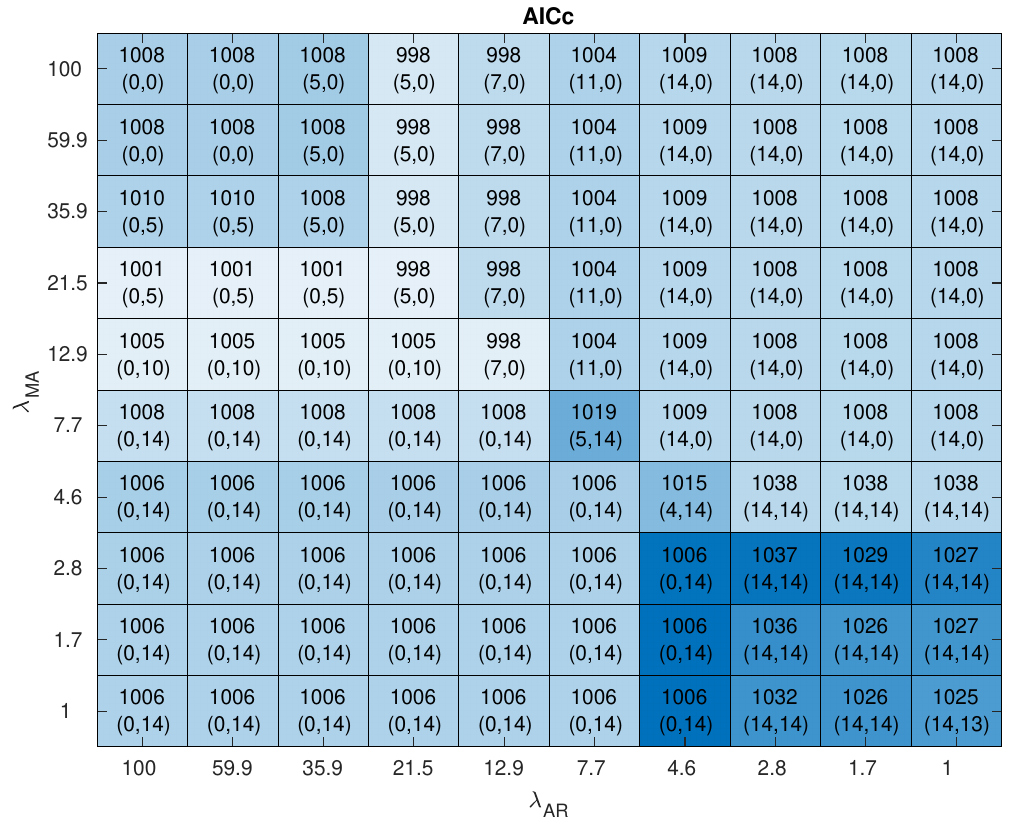}
    \end{subfigure}
    \hfill
    \begin{subfigure}[t]{0.33\textwidth}
        \centering
        \includegraphics[width=\textwidth]{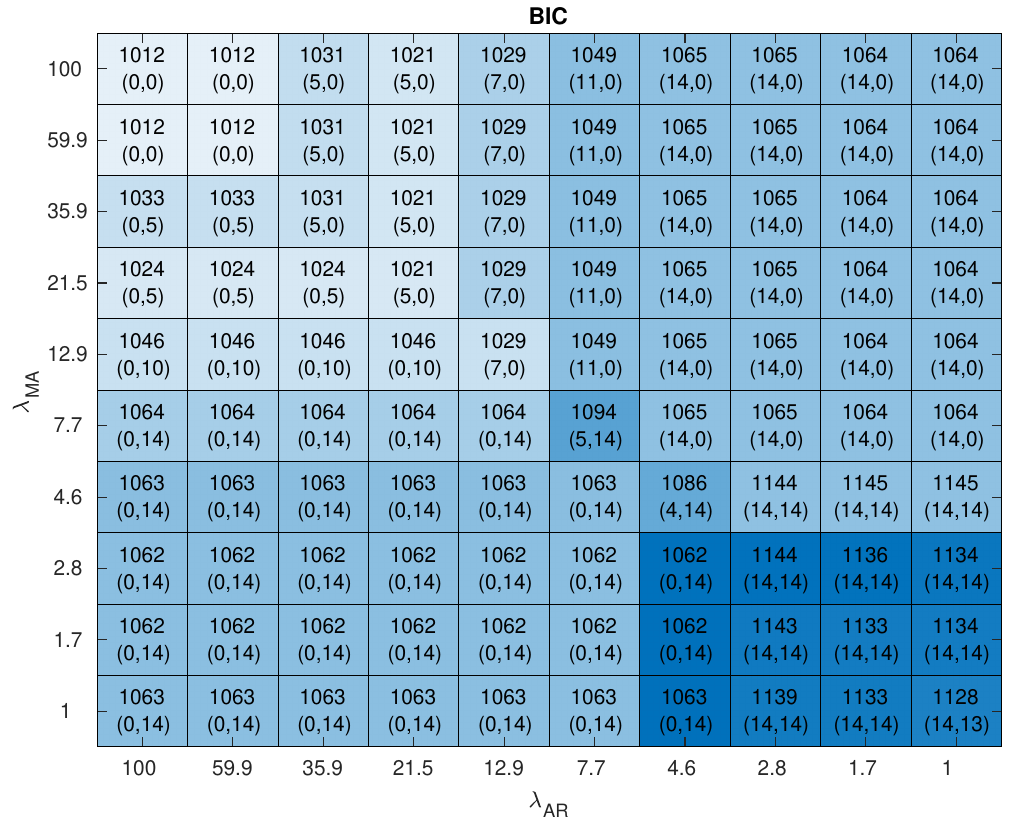}
    \end{subfigure}
    \caption{HS-ARMA fitting performance for the Netflix stock data. The identified ARMA($p,q$) models are shown inside the parentheses.}
    \label{fig:netflix}
\end{figure*}

\section{CONCLUDING REMARKS}\label{sec:conclusion}
This work presents a new learning framework that allows model identification and parameter estimation for ARMA time series models simultaneously. To do so, we use a hierarchical sparsity-inducing penalty, namely the Latent Overlapping Group (LOG) lasso, in the objective of the parameter estimation problem. While the addition of a nonsmooth (but convex) function to the objective of an already difficult nonconvex optimization seems restrictive, we propose a proximal block coordinate descent (BCD) algorithm that can solve the problem to a potential stationary point efficiently. Numerical simulation studies confirm the capabilities of the proposed learning framework to identify the true model and estimate its parameters with reasonably high accuracy. 

We believe that this study sheds some light on the hard optimization problem behind the parameter estimation of ARMA time series models (see our discussion in Section~\ref{sec:note}). Furthermore, we hope it motivates future studies to look into the convergence analysis of the proposed proximal BCD or other algorithms for such problem structures. Finally, the proposed framework can be extended to fit vector ARMA (VARMA) models where the underlying path graphs would contain multiple variables per node (see e.g. the bottom plot in Figure~\ref{fig:simple_path_graphs}), which we also leave for future studies.

\bibliographystyle{plainnat}
\bibliography{refs.bib}

\begin{thebibliography}{44}
\providecommand{\natexlab}[1]{#1}
\providecommand{\url}[1]{\texttt{#1}}
\expandafter\ifx\csname urlstyle\endcsname\relax
  \providecommand{\doi}[1]{doi: #1}\else
  \providecommand{\doi}{doi: \begingroup \urlstyle{rm}\Url}\fi

\bibitem[Adebiyi et~al.(2014)Adebiyi, Adewumi, and Ayo]{adebiyi2014comparison}
Ayodele~Ariyo Adebiyi, Aderemi~Oluyinka Adewumi, and Charles~Korede Ayo.
\newblock Comparison of arima and artificial neural networks models for stock
  price prediction.
\newblock \emph{Journal of Applied Mathematics}, 2014, 2014.

\bibitem[Bach et~al.(2012)Bach, Jenatton, Mairal, Obozinski,
  et~al.]{bach2012structured}
Francis Bach, Rodolphe Jenatton, Julien Mairal, Guillaume Obozinski, et~al.
\newblock Structured sparsity through convex optimization.
\newblock \emph{Statistical Science}, 27\penalty0 (4):\penalty0 450--468, 2012.

\bibitem[Beck and Teboulle(2009)]{beck2009fast}
Amir Beck and Marc Teboulle.
\newblock A fast iterative shrinkage-thresholding algorithm for linear inverse
  problems.
\newblock \emph{SIAM journal on imaging sciences}, 2\penalty0 (1):\penalty0
  183--202, 2009.

\bibitem[Benidir and Picinbono(1990)]{benidir1990nonconvexity}
Messaoud Benidir and B~Picinbono.
\newblock Nonconvexity of the stability domain of digital filters.
\newblock \emph{IEEE Transactions on Acoustics, Speech, and Signal Processing},
  38\penalty0 (8):\penalty0 1459--1460, 1990.

\bibitem[Bertsimas and Van~Parys(2017)]{bertsimas2017sparse}
Dimitris Bertsimas and Bart Van~Parys.
\newblock Sparse high-dimensional regression: Exact scalable algorithms and
  phase transitions.
\newblock \emph{arXiv preprint arXiv:1709.10029}, 2017.

\bibitem[Bertsimas et~al.(2016)Bertsimas, King, Mazumder,
  et~al.]{bertsimas2016best}
Dimitris Bertsimas, Angela King, Rahul Mazumder, et~al.
\newblock Best subset selection via a modern optimization lens.
\newblock \emph{The Annals of Statistics}, 44\penalty0 (2):\penalty0 813--852,
  2016.

\bibitem[Bien et~al.(2013)Bien, Taylor, and Tibshirani]{bien2013lasso}
Jacob Bien, Jonathan Taylor, and Robert Tibshirani.
\newblock A lasso for hierarchical interactions.
\newblock \emph{Annals of statistics}, 41\penalty0 (3):\penalty0 1111, 2013.

\bibitem[Billings and Yang(2006)]{billings2006application}
Daniel Billings and Jiann-Shiou Yang.
\newblock Application of the arima models to urban roadway travel time
  prediction-a case study.
\newblock In \emph{2006 IEEE International Conference on Systems, Man and
  Cybernetics}, volume~3, pages 2529--2534. IEEE, 2006.

\bibitem[Blondel et~al.(2012)Blondel, Gurbuzbalaban, Megretski, and
  Overton]{blondel2012explicit}
Vincent~D Blondel, Mert Gurbuzbalaban, Alexandre Megretski, and Michael~L
  Overton.
\newblock Explicit solutions for root optimization of a polynomial family with
  one affine constraint.
\newblock \emph{IEEE transactions on automatic control}, 57\penalty0
  (12):\penalty0 3078--3089, 2012.

\bibitem[Box et~al.(2015)Box, Jenkins, Reinsel, and Ljung]{box2015time}
George~EP Box, Gwilym~M Jenkins, Gregory~C Reinsel, and Greta~M Ljung.
\newblock \emph{Time series analysis: forecasting and control}.
\newblock John Wiley \& Sons, 2015.

\bibitem[Boyd et~al.(2011)Boyd, Parikh, Chu, Peleato, and
  Eckstein]{boyd2011distributed}
Stephen Boyd, Neal Parikh, Eric Chu, Borja Peleato, and Jonathan Eckstein.
\newblock Distributed optimization and statistical learning via the alternating
  direction method of multipliers.
\newblock \emph{Foundations and Trends{\textregistered} in Machine Learning},
  3\penalty0 (1):\penalty0 1--122, 2011.

\bibitem[Cabello et~al.(2017)Cabello, Cibulka, Kyncl, Saumell, and
  Valtr]{cabello2017peeling}
Sergio Cabello, Josef Cibulka, Jan Kyncl, Maria Saumell, and Pavel Valtr.
\newblock Peeling potatoes near-optimally in near-linear time.
\newblock \emph{SIAM Journal on Computing}, 46\penalty0 (5):\penalty0
  1574--1602, 2017.

\bibitem[Calheiros et~al.(2014)Calheiros, Masoumi, Ranjan, and
  Buyya]{calheiros2014workload}
Rodrigo~N Calheiros, Enayat Masoumi, Rajiv Ranjan, and Rajkumar Buyya.
\newblock Workload prediction using arima model and its impact on cloud
  applications? qos.
\newblock \emph{IEEE Transactions on Cloud Computing}, 3\penalty0 (4):\penalty0
  449--458, 2014.

\bibitem[Chan and Chen(2011)]{chan2011subset}
Kung-Sik Chan and Kun Chen.
\newblock Subset arma selection via the adaptive lasso.
\newblock \emph{Statistics and its Interface}, 4\penalty0 (2):\penalty0
  197--205, 2011.

\bibitem[Chang and Yap(1986)]{chang1986polynomial}
Jyun-Sheng Chang and Chee-Keng Yap.
\newblock A polynomial solution for the potato-peeling problem.
\newblock \emph{Discrete \& Computational Geometry}, 1\penalty0 (2):\penalty0
  155--182, 1986.

\bibitem[Chen et~al.(2009)Chen, Pedersen, Bak-Jensen, and Chen]{chen2009arima}
Peiyuan Chen, Troels Pedersen, Birgitte Bak-Jensen, and Zhe Chen.
\newblock Arima-based time series model of stochastic wind power generation.
\newblock \emph{IEEE transactions on power systems}, 25\penalty0 (2):\penalty0
  667--676, 2009.

\bibitem[Combettes and Trussell(1992)]{combettes1992best}
Patrick~L Combettes and H~Joel Trussell.
\newblock Best stable and invertible approximations for arma systems.
\newblock \emph{IEEE Transactions on signal processing}, 40\penalty0
  (12):\penalty0 3066--3069, 1992.

\bibitem[Del~Castillo(2002)]{del2002statistical}
Enrique Del~Castillo.
\newblock \emph{Statistical process adjustment for quality control}, volume
  369.
\newblock Wiley-Interscience, 2002.

\bibitem[Franceschi et~al.(2018)Franceschi, Frasconi, Salzo, Grazzi, and
  Pontil]{franceschi2018}
Luca Franceschi, Paolo Frasconi, Saverio Salzo, Riccardo Grazzi, and
  Massimiliano Pontil.
\newblock Bilevel programming for hyperparameter optimization and
  meta-learning.
\newblock In Jennifer Dy and Andreas Krause, editors, \emph{Proceedings of the
  35th International Conference on Machine Learning}, volume~80 of
  \emph{Proceedings of Machine Learning Research}, pages 1568--1577. PMLR,
  10--15 Jul 2018.

\bibitem[Georgiou and Lindquist(2008)]{georgiou2008convex}
Tryphon~T Georgiou and Anders Lindquist.
\newblock A convex optimization approach to arma modeling.
\newblock \emph{IEEE transactions on automatic control}, 53\penalty0
  (5):\penalty0 1108--1119, 2008.

\bibitem[Goodman(1981)]{goodman1981largest}
Jacob~E Goodman.
\newblock On the largest convex polygon contained in a non-convex n-gon, or how
  to peel a potato.
\newblock \emph{Geometriae Dedicata}, 11\penalty0 (1):\penalty0 99--106, 1981.

\bibitem[Hamilton(1994)]{hamilton1994time}
James~D Hamilton.
\newblock \emph{Time series analysis}, volume~2.
\newblock Princeton New Jersey, 1994.

\bibitem[Han et~al.(2010)Han, Wang, Zhang, et~al.]{han2010drought}
Ping Han, Peng~Xin Wang, Shu~Yu Zhang, et~al.
\newblock Drought forecasting based on the remote sensing data using arima
  models.
\newblock \emph{Mathematical and computer modelling}, 51\penalty0
  (11-12):\penalty0 1398--1403, 2010.

\bibitem[Hsu et~al.(2008)Hsu, Hung, and Chang]{hsu2008subset}
Nan-Jung Hsu, Hung-Lin Hung, and Ya-Mei Chang.
\newblock Subset selection for vector autoregressive processes using lasso.
\newblock \emph{Computational Statistics \& Data Analysis}, 52\penalty0
  (7):\penalty0 3645--3657, 2008.

\bibitem[Jacob et~al.(2009)Jacob, Obozinski, and Vert]{jacob2009group}
Laurent Jacob, Guillaume Obozinski, and Jean-Philippe Vert.
\newblock Group lasso with overlap and graph lasso.
\newblock In \emph{Proceedings of the 26th annual international conference on
  machine learning}, pages 433--440. ACM, 2009.

\bibitem[Jenatton et~al.(2011{\natexlab{a}})Jenatton, Audibert, and
  Bach]{jenatton2011structured}
Rodolphe Jenatton, Jean-Yves Audibert, and Francis Bach.
\newblock Structured variable selection with sparsity-inducing norms.
\newblock \emph{Journal of Machine Learning Research}, 12\penalty0
  (Oct):\penalty0 2777--2824, 2011{\natexlab{a}}.

\bibitem[Jenatton et~al.(2011{\natexlab{b}})Jenatton, Mairal, Obozinski, and
  Bach]{jenatton2011proximal}
Rodolphe Jenatton, Julien Mairal, Guillaume Obozinski, and Francis Bach.
\newblock Proximal methods for hierarchical sparse coding.
\newblock \emph{Journal of Machine Learning Research}, 12\penalty0
  (Jul):\penalty0 2297--2334, 2011{\natexlab{b}}.

\bibitem[Makridakis et~al.(2018)Makridakis, Spiliotis, and
  Assimakopoulos]{makridakis2018statistical}
Spyros Makridakis, Evangelos Spiliotis, and Vassilios Assimakopoulos.
\newblock Statistical and machine learning forecasting methods: Concerns and
  ways forward.
\newblock \emph{PloS one}, 13\penalty0 (3), 2018.

\bibitem[Manzour et~al.(2019)Manzour, K{\"u}{\c{c}}{\"u}kyavuz, and
  Shojaie]{manzour2019integer}
Hasan Manzour, Simge K{\"u}{\c{c}}{\"u}kyavuz, and Ali Shojaie.
\newblock Integer programming for learning directed acyclic graphs from
  continuous data.
\newblock \emph{arXiv preprint arXiv:1904.10574}, 2019.

\bibitem[Mazumder and Radchenko(2017)]{mazumder2017thediscrete}
Rahul Mazumder and Peter Radchenko.
\newblock Thediscrete dantzig selector: Estimating sparse linear models via
  mixed integer linear optimization.
\newblock \emph{IEEE Transactions on Information Theory}, 63\penalty0
  (5):\penalty0 3053--3075, 2017.

\bibitem[Moses and Liu(1991)]{moses1991determining}
Randolph~L Moses and Duixian Liu.
\newblock Determining the closest stable polynomial to an unstable one.
\newblock \emph{IEEE Transactions on signal processing}, 39\penalty0
  (4):\penalty0 901--906, 1991.

\bibitem[Nardi and Rinaldo(2011)]{nardi2011autoregressive}
Yuval Nardi and Alessandro Rinaldo.
\newblock Autoregressive process modeling via the lasso procedure.
\newblock \emph{Journal of Multivariate Analysis}, 102\penalty0 (3):\penalty0
  528--549, 2011.

\bibitem[Nesterov(2013)]{nesterov2013gradient}
Yu~Nesterov.
\newblock Gradient methods for minimizing composite functions.
\newblock \emph{Mathematical Programming}, 140\penalty0 (1):\penalty0 125--161,
  2013.

\bibitem[Nicholson et~al.(2014)Nicholson, Wilms, Bien, and
  Matteson]{nicholson2014high}
William~B Nicholson, Ines Wilms, Jacob Bien, and David~S Matteson.
\newblock High dimensional forecasting via interpretable vector autoregression.
\newblock \emph{arXiv preprint arXiv:1412.5250}, 2014.

\bibitem[Parikh et~al.(2014)Parikh, Boyd, et~al.]{parikh2014proximal}
Neal Parikh, Stephen Boyd, et~al.
\newblock Proximal algorithms.
\newblock \emph{Foundations and Trends{\textregistered} in Optimization},
  1\penalty0 (3):\penalty0 127--239, 2014.

\bibitem[Razaviyayn et~al.(2013)Razaviyayn, Hong, and
  Luo]{razaviyayn2013unified}
Meisam Razaviyayn, Mingyi Hong, and Zhi-Quan Luo.
\newblock A unified convergence analysis of block successive minimization
  methods for nonsmooth optimization.
\newblock \emph{SIAM Journal on Optimization}, 23\penalty0 (2):\penalty0
  1126--1153, 2013.

\bibitem[Ren and Zhang(2010)]{ren2010subset}
Yunwen Ren and Xinsheng Zhang.
\newblock Subset selection for vector autoregressive processes via adaptive
  lasso.
\newblock \emph{Statistics \& probability letters}, 80\penalty0
  (23-24):\penalty0 1705--1712, 2010.

\bibitem[Wang et~al.(2007)Wang, Li, and Tsai]{wang2007regression}
Hansheng Wang, Guodong Li, and Chih-Ling Tsai.
\newblock Regression coefficient and autoregressive order shrinkage and
  selection via the lasso.
\newblock \emph{Journal of the Royal Statistical Society: Series B (Statistical
  Methodology)}, 69\penalty0 (1):\penalty0 63--78, 2007.

\bibitem[Wang et~al.(2015)Wang, Chau, Xu, and Chen]{wang2015improving}
Wen-chuan Wang, Kwok-wing Chau, Dong-mei Xu, and Xiao-Yun Chen.
\newblock Improving forecasting accuracy of annual runoff time series using
  arima based on eemd decomposition.
\newblock \emph{Water Resources Management}, 29\penalty0 (8):\penalty0
  2655--2675, 2015.

\bibitem[Wilms et~al.(2017)Wilms, Basu, Bien, and Matteson]{wilms2017sparse}
Ines Wilms, Sumanta Basu, Jacob Bien, and David~S Matteson.
\newblock Sparse identification and estimation of large-scale vector
  autoregressive moving averages.
\newblock \emph{arXiv preprint arXiv:1707.09208}, 2017.

\bibitem[Yan et~al.(2017)Yan, Bien, et~al.]{yan2017hierarchical}
Xiaohan Yan, Jacob Bien, et~al.
\newblock Hierarchical sparse modeling: A choice of two group lasso
  formulations.
\newblock \emph{Statistical Science}, 32\penalty0 (4):\penalty0 531--560, 2017.

\bibitem[Zhang et~al.(2022)Zhang, Liu, and
  Davanloo~Tajbakhsh]{zhang2020firstorder}
Dewei Zhang, Yin Liu, and Sam Davanloo~Tajbakhsh.
\newblock A first-order optimization algorithm for statistical learning with
  hierarchical sparsity structure.
\newblock \emph{INFORMS Journal on Computing}, 34\penalty0 (2):\penalty0
  1126--1140, 2022.

\bibitem[Zhang et~al.(2014)Zhang, Zhang, Young, and Li]{zhang2014applications}
Xingyu Zhang, Tao Zhang, Alistair~A Young, and Xiaosong Li.
\newblock Applications and comparisons of four time series models in
  epidemiological surveillance data.
\newblock \emph{PLoS One}, 9\penalty0 (2), 2014.

\bibitem[Zhao et~al.(2009)Zhao, Rocha, and Yu]{zhao2009composite}
Peng Zhao, Guilherme Rocha, and Bin Yu.
\newblock The composite absolute penalties family for grouped and hierarchical
  variable selection.
\newblock \emph{The Annals of Statistics}, pages 3468--3497, 2009.

\end{thebibliography}

\clearpage
\newpage

\begin{appendices}

\section{EXPERIMENT RESULTS}
\setlength\tabcolsep{2pt}

\setlength{\LTcapwidth}{\textwidth}
\footnotesize
\begin{longtable}{ccrrrrrrrcrrrrrr}
\caption{ Model identification and parameter estimation accuracy of the HS-ARMA method for ten simulations from ARMA(3,2) (one realization each). Parameters are estimated using the proximal BCD Algorithm \ref{alg:master_alg}. Boldface columns denote the best identified models with the lowest estimation errors.}\label{tab:sample_detail} \\

\toprule  
\endfirsthead

\multicolumn{16}{c}%
{\tablename\ \thetable{} -- continued from previous page} \\
\hline 
\endhead

\hline \multicolumn{3}{r}{{Continued on next page}} \\ \hline
\endfoot

\endlastfoot
                             &                       & \multicolumn{6}{c}{$\lambda_0$} &                &                & \multicolumn{6}{c}{$\lambda_0$}                                                                                                                         \\
                             & $(\phi^{*,1},\theta^{*,1})$ & 0.5                             & 1              & 2              & 3                               & 5              & 10    &  & $(\phi^{*,2},\theta^{*,2})$   & 0.5   & 1     & 2              & 3              & 5     & 10    \\
                                 \cmidrule(r){3-8}  \cmidrule(r){11-16}
 $\phi_1$               & -0.16                 & -0.23                           & -0.22          & -0.25          & \textbf{-0.10}                  & 0.01           & 0.05  &  & 0.13                    & 0.24  & 0.18  & \textbf{0.10}  & 0.08           & 0.04  & -0.14 \\
      $\phi_2$               & -0.98                 & -0.53                           & -0.60          & -0.83          & \textbf{-0.99}                  & -0.99          & -0.99 &  & 0.42                    & 0.41  & 0.44  & \textbf{0.45}  & 0.46           & 0.48  & 0.53  \\
      $\phi_3$               & -0.22                 & -0.21                           & -0.24          & -0.32          & \textbf{-0.16}                  & -0.05          & -0.01 &  & -0.44                   & -0.51 & -0.50 & \textbf{-0.42} & -0.39          & -0.34 & -0.15 \\
      $\phi_4$               & 0.00                  & 0.44                            & 0.38           & 0.15           & \textbf{0.00}                   & 0.00           & 0.00  &  & 0.00                    & 0.04  & 0.01  & \textbf{0.00}  & 0.00           & 0.00  & 0.00  \\
      $\phi_5$               & 0.00                  & 0.10                            & 0.06           & 0.00           & \textbf{0.00}                   & 0.00           & 0.00  &  & 0.00                    & 0.00  & 0.00  & \textbf{0.00}  & 0.00           & 0.00  & 0.00  \\
      $\theta_1$             & -0.45                 & -0.38                           & -0.38          & -0.35          & \textbf{-0.46}                  & -0.51          & -0.50 &  & 0.49                    & 0.37  & 0.43  & \textbf{0.50}  & 0.51           & 0.55  & 0.68  \\
      $\theta_2$             & 0.91                  & 0.38                            & 0.44           & 0.64           & \textbf{0.89}                   & 0.92           & 0.85  &  & 0.34                    & 0.28  & 0.29  & \textbf{0.31}  & 0.30           & 0.28  & 0.26  \\
      $\theta_3$             & 0.00                  & 0.29                            & 0.27           & 0.20           & \textbf{0.00}                   & -0.03          & 0.00  &  & 0.00                    & -0.03 & 0.00  & \textbf{0.00}  & 0.00           & 0.00  & 0.00  \\
      $\theta_4$             & 0.00                  & -0.45                           & -0.40          & -0.20          & \textbf{0.00}                   & 0.00           & 0.00  &  & 0.00                    & 0.01  & 0.00  & \textbf{0.00}  & 0.00           & 0.00  & 0.00  \\
      $\theta_5$             & 0.00                  & 0.00                            & 0.00           & 0.00           & \textbf{0.00}                   & 0.00           & 0.00  &  & 0.00                    & 0.00  & 0.00  & \textbf{0.00}  & 0.00           & 0.00  & 0.00  \\
      \hline
      $\epsilon_{\lambda_0}$ &                       & 0.98                            & 0.82           & 0.35           & \textbf{0.10}                   & 0.21           & 0.31  &  &                         & 0.20  & 0.13  & \textbf{0.06}  & 0.09           & 0.16  & 0.37  \\ 
           \hline \hline
                             &                       & \multicolumn{6}{c}{$\lambda_0$} &                &                & \multicolumn{6}{c}{$\lambda_0$}                                                                                                                         \\
                             & $(\phi^{*,3},\theta^{*,3})$ & 0.5                             & 1              & 2              & 3                               & 5              & 10    &  & $(\phi^{*,4},\theta^{*,4})$   & 0.5   & 1     & 2              & 3              & 5     & 10    \\
                                                              \cmidrule(r){3-8}  \cmidrule(r){11-16}

      $\phi_1$               & -0.64                 & -0.48                           & -0.46          & -0.47          & -0.55                           & \textbf{-0.66} & -0.75 &  & -0.88                   & -0.12 & -0.18 & -0.37          & \textbf{-0.51} & -0.15 & -0.16 \\
      $\phi_2$               & -0.70                 & -0.25                           & -0.33          & -0.46          & -0.57                           & \textbf{-0.61} & -0.40 &  & -0.28                   & 0.13  & 0.18  & 0.18           & \textbf{0.07}  & 0.28  & 0.24  \\
      $\phi_3$               & -0.56                 & -0.25                           & -0.29          & -0.37          & -0.46                           & \textbf{-0.45} & -0.24 &  & 0.37                    & 0.35  & 0.41  & 0.47           & \textbf{0.43}  & 0.29  & 0.27  \\
      $\phi_4$               & 0.00                  & 0.33                            & 0.26           & 0.14           & 0.04                            & \textbf{0.00}  & 0.01  &  & 0.00                    & -0.39 & -0.34 & -0.18          & \textbf{-0.11} & -0.21 & -0.07 \\
      $\phi_5$               & 0.00                  & 0.18                            & 0.11           & 0.00           & 0.00                            & \textbf{0.00}  & 0.00  &  & 0.00                    & 0.07  & 0.03  & 0.00           & \textbf{0.00}  & 0.00  & 0.00  \\
      $\theta_1$             & -0.49                 & -0.61                           & -0.63          & -0.62          & -0.55                           & \textbf{-0.44} & -0.28 &  & 0.84                    & 0.09  & 0.14  & 0.32           & \textbf{0.46}  & 0.03  & 0.00  \\
      $\theta_2$             & 0.49                  & 0.20                            & 0.29           & 0.39           & 0.41                            & \textbf{0.30}  & 0.00  &  & 0.57                    & 0.19  & 0.14  & 0.14           & \textbf{0.24}  & 0.00  & 0.00  \\
      $\theta_3$             & 0.00                  & 0.11                            & 0.07           & 0.00           & 0.00                            & \textbf{0.00}  & 0.01  &  & 0.00                    & -0.19 & -0.24 & -0.19          & \textbf{-0.06} & 0.00  & 0.00  \\
      $\theta_4$             & 0.00                  & -0.15                           & -0.09          & 0.00           & 0.00                            & \textbf{0.00}  & 0.01  &  & 0.00                    & 0.17  & 0.10  & 0.00           & \textbf{0.00}  & 0.00  & 0.00  \\
      $\theta_5$             & 0.00                  & 0.00                            & 0.00           & 0.00           & 0.00                            & \textbf{0.00}  & 0.00  &  & 0.00                    & 0.01  & 0.00  & 0.00           & \textbf{0.00}  & 0.00  & 0.00  \\
      \hline
      $\epsilon_{\lambda_0}$ &                       & 0.76                            & 0.61           & 0.46           & 0.57                            & \textbf{0.30}  & 0.81  &  &                         & 1.29  & 1.24  & 0.98           & \textbf{0.74}  & 1.41  & 1.36  \\

      \hline \hline
                             &                       & \multicolumn{6}{c}{$\lambda_0$} &                &                & \multicolumn{6}{c}{$\lambda_0$}                                                                                                                         \\
                             & $(\phi^{*,5},\theta^{*,5})$ & 0.5                             & 1              & 2              & 3                               & 5              & 10    &  & $(\phi^{*,6},\theta^{*,6})$   & 0.5   & 1     & 2              & 3              & 5     & 10    \\
                                                              \cmidrule(r){3-8}  \cmidrule(r){11-16}

      $\phi_1$               & -0.19                 & -0.32                           & \textbf{-0.37} & -0.37          & -0.51                           & -0.53          & -0.53 &  & -0.58                   & -0.21 & -0.25 & -0.43          & \textbf{-0.57} & -0.57 & -0.53 \\
      $\phi_2$               & 0.55                  & 0.04                            & \textbf{0.06}  & -0.12          & -0.20                           & -0.21          & -0.20 &  & 0.61                    & 0.57  & 0.67  & 0.67           & \textbf{0.60}  & 0.60  & 0.63  \\
      $\phi_3$               & 0.52                  & 0.05                            & \textbf{0.13}  & 0.16           & 0.14                            & 0.13           & 0.01  &  & 0.85                    & 0.48  & 0.58  & 0.74           & \textbf{0.82}  & 0.81  & 0.78  \\
      $\phi_4$               & 0.00                  & -0.16                           & \textbf{-0.11} & -0.10          & -0.06                           & -0.01          & 0.00  &  & 0.00                    & -0.17 & -0.20 & -0.12          & \textbf{0.00}  & 0.00  & -0.01 \\
      $\phi_5$               & 0.00                  & 0.00                            & \textbf{0.00}  & 0.01           & 0.00                            & 0.00           & 0.00  &  & 0.00                    & 0.19  & 0.09  & 0.00           & \textbf{0.00}  & 0.00  & 0.00  \\
      $\theta_1$             & -0.30                 & -0.17                           & \textbf{-0.12} & -0.13          & 0.00                            & 0.00           & 0.00  &  & 0.24                    & -0.16 & -0.12 & 0.05           & \textbf{0.19}  & 0.18  & 0.08  \\
      $\theta_2$             & -0.64                 & -0.18                           & \textbf{-0.24} & -0.01          & 0.00                            & 0.00           & 0.00  &  & 0.23                    & 0.40  & 0.29  & 0.21           & \textbf{0.18}  & 0.13  & 0.01  \\
      $\theta_3$             & 0.00                  & 0.22                            & \textbf{0.16}  & 0.00           & 0.00                            & 0.00           & 0.00  &  & 0.00                    & -0.06 & -0.05 & -0.01          & \textbf{0.00}  & 0.00  & 0.00  \\
      $\theta_4$             & 0.00                  & -0.03                           & \textbf{-0.02} & 0.00           & 0.00                            & 0.00           & 0.00  &  & 0.00                    & 0.04  & 0.00  & 0.00           & \textbf{0.00}  & 0.00  & 0.00  \\
      $\theta_5$             & 0.00                  & 0.00                            & \textbf{0.00}  & 0.00           & 0.00                            & 0.00           & 0.00  &  & 0.00                    & 0.00  & 0.00  & 0.00           & \textbf{0.00}  & 0.00  & 0.00  \\
      \hline
      $\epsilon_{\lambda_0}$ &                       & 0.89                            & \textbf{0.77}  & 1.11           & 1.15                            & 1.16           & 1.20  &  &                         & 0.73  & 0.60  & 0.26           & \textbf{0.09}  & 0.14  & 0.32  \\
      \hline \hline

                             &                       & \multicolumn{6}{c}{$\lambda_0$} &                &                & \multicolumn{6}{c}{$\lambda_0$}                                                                                                                         \\
                             & $(\phi^{*,7},\theta^{*,7})$ & 0.5                             & 1              & 2              & 3                               & 5              & 10    &  & $(\phi^{*,8},\theta^{*,8})$   & 0.5   & 1     & 2              & 3              & 5     & 10    \\
                                                              \cmidrule(r){3-8}  \cmidrule(r){11-16}

      $\phi_1$               & -0.34                 & 0.03                            & \textbf{-0.02} & 0.05           & 0.24                            & 0.41           & 0.39  &  & 0.92                    & 0.83  & 0.83  & \textbf{0.88}  & 0.93           & 0.93  & 0.93  \\
      $\phi_2$               & -0.53                 & -0.54                           & \textbf{-0.61} & -0.68          & -0.78                           & -0.84          & -0.82 &  & 0.91                    & 0.70  & 0.78  & \textbf{0.96}  & 0.91           & 0.90  & 0.89  \\
      $\phi_3$               & -0.65                 & -0.39                           & \textbf{-0.40} & -0.31          & -0.13                           & 0.00           & 0.00  &  & -0.88                   & -0.50 & -0.58 & \textbf{-0.83} & -0.87          & -0.87 & -0.86 \\
      $\phi_4$               & 0.00                  & 0.13                            & \textbf{0.05}  & 0.00           & 0.00                            & 0.00           & 0.00  &  & 0.00                    & 0.19  & 0.12  & \textbf{-0.04} & 0.00           & 0.00  & -0.01 \\
      $\phi_5$               & 0.00                  & 0.00                            & \textbf{0.00}  & 0.00           & 0.00                            & 0.00           & 0.00  &  & 0.00                    & -0.27 & -0.20 & \textbf{0.00}  & 0.00           & 0.00  & -0.01 \\
      $\theta_1$             & 0.32                  & -0.03                           & \textbf{0.01}  & -0.04          & -0.24                           & -0.47          & -0.37 &  & 0.73                    & 0.81  & 0.81  & \textbf{0.76}  & 0.67           & 0.64  & 0.57  \\
      $\theta_2$             & -0.49                 & -0.47                           & \textbf{-0.42} & -0.34          & -0.19                           & -0.02          & -0.01 &  & 0.21                    & 0.58  & 0.50  & \textbf{0.21}  & 0.13           & 0.10  & 0.01  \\
      $\theta_3$             & 0.00                  & 0.12                            & \textbf{0.09}  & 0.05           & 0.00                            & 0.00           & 0.00  &  & 0.00                    & 0.25  & 0.18  & \textbf{0.00}  & 0.00           & 0.00  & 0.00  \\
      $\theta_4$             & 0.00                  & -0.05                           & \textbf{-0.02} & -0.01          & 0.00                            & 0.00           & 0.00  &  & 0.00                    & 0.05  & 0.03  & \textbf{0.00}  & 0.00           & 0.00  & 0.00  \\
      $\theta_5$             & 0.00                  & 0.01                            & \textbf{0.00}  & 0.00           & 0.00                            & 0.00           & 0.00  &  & 0.00                    & 0.00  & 0.00  & \textbf{0.00}  & 0.00           & 0.00  & 0.00  \\
      \hline
      $\epsilon_{\lambda_0}$ &                       & 0.61                            & \textbf{0.53}  & 0.64           & 1.19                            & 1.39           & 1.29  &  &                         & 0.72  & 0.53  & \textbf{0.05}  & 0.11           & 0.16  & 0.25  \\
      \hline \hline
                             &                       & \multicolumn{6}{c}{$\lambda_0$} &                &                & \multicolumn{6}{c}{$\lambda_0$}                                                                                                                         \\
                             & $(\phi^{*,9},\theta^{*,9})$ & 0.5                             & 1              & 2              & 3                               & 5              & 10    &  & $(\phi^{*,10},\theta^{*,10})$ & 0.5   & 1     & 2              & 3              & 5     & 10    \\
                                                              \cmidrule(r){3-8}  \cmidrule(r){11-16}

      $\phi_1$               & -0.76                 & -0.66                           & -0.64          & \textbf{-0.75} & -0.81                           & -0.83          & -0.89 &  & 0.67                    & 0.21  & 0.26  & 0.37           & \textbf{0.56}  & 0.81  & 1.04  \\
      $\phi_2$               & -0.46                 & -0.22                           & -0.29          & \textbf{-0.44} & -0.45                           & -0.42          & -0.41 &  & 0.20                    & 0.16  & 0.22  & 0.34           & \textbf{0.30}  & 0.00  & -0.33 \\
      $\phi_3$               & -0.59                 & -0.38                           & -0.45          & \textbf{-0.58} & -0.55                           & -0.49          & -0.42 &  & -0.41                   & -0.05 & -0.10 & -0.30          & \textbf{-0.44} & -0.29 & -0.09 \\
      $\phi_4$               & 0.00                  & 0.19                            & 0.15           & \textbf{0.00}  & 0.00                            & 0.00           & 0.00  &  & 0.00                    & -0.18 & -0.21 & -0.14          & \textbf{-0.01} & 0.00  & 0.00  \\
      $\phi_5$               & 0.00                  & 0.11                            & 0.05           & \textbf{0.00}  & 0.00                            & 0.00           & 0.00  &  & 0.00                    & -0.15 & -0.09 & 0.00           & \textbf{0.00}  & 0.00  & 0.00  \\
      $\theta_1$             & -0.84                 & -0.90                           & -0.92          & \textbf{-0.82} & -0.76                           & -0.74          & -0.62 &  & -0.57                   & -0.13 & -0.18 & -0.29          & \textbf{-0.47} & -0.74 & -0.90 \\
      $\theta_2$             & 0.17                  & 0.06                            & 0.15           & \textbf{0.15}  & 0.09                            & 0.03           & 0.00  &  & -0.30                   & -0.23 & -0.30 & -0.43          & \textbf{-0.40} & -0.14 & 0.00  \\
      $\theta_3$             & 0.00                  & 0.09                            & 0.03           & \textbf{0.00}  & 0.00                            & 0.01           & 0.00  &  & 0.00                    & -0.37 & -0.32 & -0.12          & \textbf{0.00}  & 0.00  & -0.01 \\
      $\theta_4$             & 0.00                  & 0.00                            & 0.00           & \textbf{0.00}  & 0.00                            & 0.01           & 0.00  &  & 0.00                    & -0.05 & 0.00  & 0.00           & \textbf{0.00}  & 0.00  & 0.00  \\
      $\theta_5$             & 0.00                  & 0.00                            & 0.00           & \textbf{0.00}  & 0.00                            & 0.00           & 0.00  &  & 0.00                    & 0.00  & 0.00  & 0.00           & \textbf{0.00}  & 0.00  & 0.00  \\
      \hline
      $\epsilon_{\lambda_0}$ &                       & 0.42                            & 0.30           & \textbf{0.05}  & 0.15                            & 0.25           & 0.36  &  &                         & 0.85  & 0.72  & 0.45           & \textbf{0.12}  & 0.48  & 0.78  \\

\bottomrule
\end{longtable}
\end{appendices}

\end{document}